\documentclass[12pt]{article}
\usepackage{amssymb,graphicx}
\usepackage{tabularx}
\usepackage{array}
\usepackage{algorithmic}

\def\squarebox#1{\hbox to #1{\hfill\vbox to #1{\vfill}}}
\newcommand{\qed}{\hspace*{\fill}
            \vbox{\hrule\hbox{\vrule\squarebox{.667em}\vrule}\hrule}\smallskip\newline}
\newtheorem{THEOREM}{Theorem}
\newenvironment{theorem}{\begin{THEOREM} \hspace{-.85em} {\bf :} \rm}                        {\end{THEOREM}}
\newtheorem{LEMMA}{Lemma}
\newenvironment{lemma}{\begin{LEMMA} \hspace{-.85em} {\bf :} \rm}                      {\end{LEMMA}}
\newtheorem{COROLLARY}{Corollary}
\newenvironment{corollary}{\begin{COROLLARY} \hspace{-.85em} {\bf :} \rm}                          {\end{COROLLARY}}
\newenvironment{proof}{\noindent {\bf Proof:}}{\qed}
\newtheorem{DEFINITION}{Definition}
\newenvironment{definition}{\begin{DEFINITION} \hspace{-.85em} {\bf :} \rm}
                            {\end{DEFINITION}}
\newtheorem{CLAIM}{Claim}
                      {\end{CLAIM}}

\begin{document}

\setlength{\baselineskip}{18pt}

\title{Algorithms for Locating Constrained Optimal Intervals\footnote{An earlier version of the second part of this work appeared in \textit{Proceedings of the 18th International Symposium on Algorithms and
Computation}, Japan, 2007.}}

\author{Hsiao-Fei Liu$^{1}$, Peng-An Chen$^1$, and Kun-Mao Chao$^{1,2,3}$\\
\\
$^1$Department of Computer Science and Information Engineering \\
$^2$Graduate Institute of Biomedical Electronics and
Bioinformatics \\
$^3$Graduate Institute of Networking and Multimedia \\
National Taiwan University, Taipei, Taiwan 106
}\maketitle

\begin{abstract}
In this work, we obtain the following new results.
\begin{itemize}
\item Given a sequence $D=((h_1,s_1), (h_2,s_2) \ldots, (h_n,s_n))$ of number pairs, where $s_i>0$ for all $i$, and a number $L_h$, we propose an $O(n)$-time algorithm for
finding an index interval $[i,j]$ that maximizes
$\frac{\sum_{k=i}^{j} h_k}{\sum_{k=i}^{j} s_k}$ subject to
$\sum_{k=i}^{j} h_k \geq L_h$.

\item  Given a sequence $D=((h_1,s_1), (h_2,s_2) \ldots, (h_n,s_n))$ of number pairs, where $s_i=1$ for all $i$, and an integer $L_s$ with $1\leq L_s\leq n$,
we propose an $O(n\frac{T(L_s^{1/2})}{L_s^{1/2}})$-time algorithm
for finding an index interval $[i,j]$ that maximizes
$\frac{\sum_{k=i}^{j} h_k}{\sqrt{\sum_{k=i}^{j} s_k}}$ subject to
$\sum_{k=i}^{j} s_k \geq L_s$, where $T(n')$ is the time required to
solve the all-pairs shortest paths problem on a graph of $n'$ nodes.
By the latest result of Chan \cite{Chan}, $T(n')=O(n'^3
\frac{(\log\log n')^3}{(\log n')^2})$, so our algorithm runs in
subquadratic time $O(nL_s\frac{(\log\log L_s)^3}{(\log L_s)^2})$.
\end{itemize}

\end{abstract}

\thispagestyle{plain}

\section {Introduction}
 Given a sequence $D=((h_1,s_1), (h_2,s_2) \ldots, (h_n,s_n))$ of number pairs, where $s_i>0$ for all $i$, define the $support$, $hit$-$support$, $confidence$, $eccentricity$,
 and $aberrance$ of an index interval $I=[i,j]=\{i,i+1,\ldots,j\}$ to be $\sum_{k=i}^js_k$, $\sum_{k=i}^jh_k$, $\frac{\sum_{k=i}^jh_k}{\sum_{k=i}^js_k}$, $\frac{\sum_{k=i}^jh_k}{\sqrt{\sum_{k=i}^js_k}}$, and  $\frac{|\sum_{k=i}^jh_k|}{\sqrt{\sum_{k=i}^js_k}},$ respectively.
 Denote by $sup(i,j)$, $hit(i,j)$, $conf(i,j)$, $ecc(i,j),$ and $aberr(i,j)$ the support, hit-support, confidence, eccentricity, and aberrance of index interval $I=[i,j]$, respectively. The sequence $D$ is said to be $plain$ if and only if $s_i=1$ for all $i$.
 An index interval $I=[i,j]$ is said to be $amble$ with respect to a support lower bound $L_s$ if $sup(i,j)\geq L_s$. An index interval $I=[i,j]$ is said to be $endorsed$ with respect to a hit-support lower bound $L_h$ if $hit(i,j)\geq L_h$. An index interval $I=[i,j]$ is said to be $confident$ with respect to a confidence lower bound $L_c$ if $conf(i,j)\geq L_c$.
 Consider the following problems arising in association rule mining~\cite{Fukuda99,Fukuda01},
 computational
 biology~\cite{Allison03,Bernholt2007,Chen05,Cheng08,Chung,Fan,Goldwasser,Huang,Kim,Lee07,Lin,Wang03},
 and statistics~\cite{Davies}.

\begin{itemize}

          \item {\sc Hit-Constrained Max Confidence Interval (HCI) Problem}: Given a sequence $D=((h_1,s_1), (h_2,s_2), \ldots, (h_n,s_n))$ of number pairs, where $s_i>0$ for all $i$, and a hit-support lower bound $L_h$, find an endorsed interval $I=[i,j]$ maximizing the confidence $conf(i,j).$ Bernholt~$et~al.$~\cite{Bernholt2007}'s results imply an $O(n\log n)$-time algorithm for this problem, and we give an $O(n)$-time algorithm in this paper.

           \item {\sc Plain Support-Constrained Max Eccentricity Interval (PSEI) Problem}: Given a plain sequence $D=((h_1,s_1), (h_2,s_2), \ldots, (h_n,s_n))$ of number pairs, where $s_i=1$ for all $i$, and a support lower bound $L_s$ with $1\leq L_s \leq n$, find an amble interval $I=[i,j]$ maximizing the eccentricity $ecc(i,j).$  Lipson $et~al.$~\cite{Lipson}
            proposed an approximation scheme for the case $L_s=1$. Specifically, given an $\epsilon \in (0,1/5]$, their algorithm guarantees to outputs an index interval $[i,j]$ such that $ecc(i,j)$ is at least Opt$/\alpha(\epsilon)$ in $O(n\epsilon^{-2})$ time, where Opt = $\max\{ecc(i',j'): 1\leq i' \leq j'\leq  n\}$ and $\alpha(\epsilon) = (1-\sqrt{2\epsilon(2+\epsilon)})^{-1}$. In this paper, we propose an $O(n\frac{T(L_s^{1/2})}{L_s^{1/2}})$-time algorithm for this problem, where $T(n')$ is the time required to solve the all-pairs shortest paths
            problem on a graph of $n'$ nodes. By the latest result of Chan
            \cite{Chan}, $T(n')=O(n'^3 \frac{(\log\log n')^3}{(\log n')^2})$, so
            our algorithm runs in subquadratic time $O(nL_s\frac{(\log\log L_s)^3}{(\log L_s)^2})$. To the best of our knowledge, it is the first subquadratic result for this problem.

          \item {\sc Confidence-Constrained Max Hit Interval (CHI) Problem}: Given a sequence $D=((h_1,s_1), (h_2,s_2), \ldots, (h_n,s_n))$ of number pairs, where $s_i>0$ for all $i$, and a confidence lower bound $L_c$, find a confident interval $I=[i,j]$ maximizing the hit-support $hit(i,j).$ Bernholt~$et~al.$~\cite{Bernholt2007}'s results imply an $O(n\log n)$-time algorithm for this problem, and recently, Cheng~$et~al.$~\cite{Cheng08} obtained an $O(n)$-time algorithm.

          \item {\sc Support-Constrained Max Confidence Interval (SCI) Problem}: Given a sequence $D=((h_1,s_1), (h_2,s_2), \ldots, (h_n,s_n))$ of number pairs, where $s_i>0$ for all $i$, and a support lower bound $L_s$, find an amble interval $I=[i,j]$ maximizing the confidence $conf(i,j).$ This problem was studied in~\cite{Bernholt2007,Chung,Fukuda99,Goldwasser,Huang,Kim,Lee07,Lin} and can be solved in $O(n)$ time~\cite{Bernholt2007,Chung,Fukuda99,Goldwasser,Lee07}.

          \item {\sc Confidence-Constrained Max Support Interval (CSI) Problem}: Given a sequence $D=((h_1,s_1), (h_2,s_2), \ldots, (h_n,s_n))$ of number pairs, where $s_i>0$ for all $i$, and a confidence lower bound $L_c$, find
            a confident interval $I=[i,j]$ maximizing the support $sup(i,j).$ This problem was studied in~\cite{Allison03,Bernholt2007,Chen05,Fukuda99,Wang03} and can be solved in $O(n)$ time~\cite{Fukuda99}.

          \item {\sc Support-Constrained Max Aberrance Interval (SAI) Problem}: Given a sequence $D=((h_1,s_1), (h_2,s_2), \ldots, (h_n,s_n))$ of number pairs, where $s_i>0$ for all $i$, a support lower bound $L_s$, and a support upper bound $U_s$, find
             an index interval $I=[i,j]$ maximizing the aberrance $aberr(i,j)$ subject to $L_s\leq sup(i,j)\leq U_s.$ Bernholt~$et~al.$~\cite{Bernholt2007} proposed an $O(n)$-time for this problem.

          \item {\sc Support-Constrained Max Hit Interval (SHI) Problem}: Given a sequence $D=((h_1,s_1), (h_2,s_2), \ldots, (h_n,s_n))$ of number pairs, where $s_i>0$ for all $i$, and a support lower bound $L_s$, find
            an amble interval $I=[i,j]$ maximizing the hit-support $hit(i,j).$ This problem was solvable in $O(n)$ time by algorithms in~\cite{Bernholt2007,Fan,Lin}.

\end{itemize}

Results for these problems are summarized in Table~\ref{Results on
problems of finding constrained optimal intervals}. The rest of this
paper is organized as follows. In Section~2, we give a linear-time
algorithm for the HCI problem, which is an adaption of the algorithm
by Chung and Lu~\cite{Chung}. In Section~3, we give the first
subquadratic time algorithm for the PSEI problem. Finally, we close
the paper by mentioning a few open problems.

\begin{table}[h!]\label{Results on problems of finding constrained optimal intervals}
\centering
\caption{Results on problems of finding constrained optimal intervals.}
\begin{tabular}{|l|l|l|}
\multicolumn{3}{c}{}\\
\hline
Paper&Problem&
Time\\
\hline
Bernholt~$et~al.$~\cite{Bernholt2007}$\ \checkmark$ & HCI & $O(n\log n)$ \\
This paper & HCI & $O(n)$ \\
\hline
Lipson~$et~al.~\cite{Lipson}$ & PSEI & Approximation Scheme for $L_s=1$\\
This paper & PSEI & $O(nL_s\frac{(\log\log
L_s)^3}{(\log L_s)^2})$  \\
\hline
Bernholt~$et~al.$~\cite{Bernholt2007}$\ \checkmark$ & CHI & $O(n\log n)$ \\
Cheng~$et~al.$~\cite{Cheng08} & CHI & $O(n)$ \\
\hline
Huang~\cite{Huang}& SCI & $O(nL_s)^{*}$ \\
Fukuda~$et~al.$~\cite{Fukuda99}& SCI & $O(n)$  \\
Lin~$et~al.$~\cite{Lin} & SCI & $O(n\log L_s)^{*}$ \\
Kim~$et~al.$~\cite{Kim}& SCI & $O(n)$ \\
Chung \& Lu~\cite{Chung}$\ \checkmark$ & SCI & $O(n)$ \\
Goldwasser~$et~al.$~\cite{Goldwasser}$\ \checkmark$ & SCI & $O(n)$ \\
Bernholt~$et~al.$~\cite{Bernholt2007}$\ \checkmark$ & SCI & $O(n)$ \\
Lee~$et~al.$~\cite{Lee07}$\ \checkmark$ & SCI & $O(n)$ \\
\hline
Fukuda~$et~al.$~\cite{Fukuda99}&CSI & $O(n)$  \\
Allison~\cite{Allison03} &CSI& $O(n^2)^*$ \\
Wang \& Xu~\cite{Wang03} &CSI& $O(n)$ \\
Chen \& Chao~\cite{Chen05}$\ \checkmark$ & CSI &$O(n)^{*}$ \\
Bernholt~$et~al.$~\cite{Bernholt2007}$\ \checkmark$ & CSI & $O(n)$ if $h_i\in \{0,1\}$ \& $s_i=1$ for all $i$ \\
\hline
Bernholt~$et~al.$~\cite{Bernholt2007}$\ \checkmark$ & SAI & $O(n)$  \\
\hline
Lin~\cite{Lin} &SHI& $O(n)$ \\
Fan~\cite{Fan} &SHI& $O(n)$ \\
Bernholt~$et~al.$~\cite{Bernholt2007}$\ \checkmark$ & SHI & $O(n)$  \\
\hline
\end{tabular}
\bigskip

\begin{tabular}{l}
$\checkmark$As a matter of fact, they solved more general problems than we list.\\
$^{*}$ The time bounds hold provided that the input sequence $D$ is plain.\\
\end{tabular}
\end{table}


\section{A linear time algorithm for the HCI problem}
For ease of exposition, we assume $L_h\geq0$ in subsequent
discussion. The restriction can be overcame as follows. If $L_h < 0$
and $h_i\geq 0$ for some $i$, then it is safe to reset $L_h$ to~0.
Otherwise, if $L_h < 0$ and $h_i<0$ for all $i$, let
$D'=((h'_1,s'_1),(h'_2,s'_2),\ldots,(h'_n,s'_n))=
((s_1,-h_1),(s_2,-h_2),\ldots,(s_n,-h_n))$ and $U_s = -L_h$. The
problem is then reduced to finding an index interval $I=[i,j]$ that
maximizes $\frac{\sum_{k=i}^jh'_k}{\sum_{k=i}^js'_k}$ subject to
$\sum_{k=i}^js'_k\leq U_s,$ which is solvable in $O(n)$ time in an
online manner\footnote{An algorithm is said to run in an online
manner if and only if it can process its input piece-by-piece and
maintain a solution for the pieces processed so far.} by Chung and
Lu's algorithm~\cite{Chung}.

\subsection {Preliminaries}
Let $H=(h_1,h_2,\ldots,h_n)$ and $P_H[0..n]$ be the prefix-sum array
of $H$, where $P_H[0]=0$ and $P_H[i]=P_H[i-1]+h_i$ for each
$i=1,2,\ldots,n$. Let $S=(s_1,s_2,\ldots,s_n)$ and $P_S[0..n]$ be
the prefix-sum array of $S$, where $P_S[0]=0$ and
$P_S[i]=P_S[i-1]+s_i$ for each $i=1,2,\ldots,n$. Both $P_H$ and $P_S$ can be computed in $O(n)$ time in an online manner. Note that
$hit(i,j)=P_H[i]-P_H[j-1]$, $sup(i,j)=P_S[i]-P_s[j-1]$, and
$conf(i,j) = \frac{P_H[i]-P_H[j-1]}{P_S[i]-P_s[j-1]}$. Therefore,
by keeping $P_H$ and $P_S$, each computation
of the hit-support, support, or confidence of an index interval
can be done in constant time.  We next introduce the notion of partners. For
technical reasons, we define $hit(0,p)=L_s$ and $conf(0,p)=-\infty$
for all indices~$p$.

\begin{definition}
Given an index $q$, an nonnegative integer $p$ is said to be a \textit{partner} of $q$ if and only if $hit(p,q)\geq L_h$.
\end{definition}

\begin{definition}
Given an index $q$, an integer $p$ is said to be the \textit{best partner} $\pi_q$ of $q$ if and only if $p$ is the largest partner of $q$ such that $conf(p,q)=\max\{conf(i,q): i \mbox{ is a partner of }q\}.$
\end{definition}

\begin{definition}
Denote by $r_q$ the right most partner (i.e.,
the largest partner) of index $q$. Define the \textit{ideal right most partner} of
index $q$ as $\hat{r}_q=\max\limits_{1\leq p\leq q}r_p$.
\end{definition}

\begin{definition}
An index $q$ is said to be a \textit{good index} if and only if $\hat{r}_q = r_q$.
\end{definition}

Let $q^*$ be an index that maximizes $conf(\pi_{q^*},q^*)$. If
$\pi_{q^*}=0$, then there is no endorsed interval; otherwise, index
interval $[\pi_{q^*},q^*]$ is a maximum-confidence endorsed
interval. Therefore, to solve the HCI problem, it suffices to find
index $q^*$. The next lemma enable us to eliminate bad indices from
consideration.

\begin{lemma}\label{real < ideal}
If $q$ is a bad index, then index interval $[\pi_q,q]$ is not a
maximum-confidence endorsed interval.
\end{lemma}
\begin{proof}
If $q$ is a bad index, then there exist $p<q$ such that $r_q<
\hat{r}_q =r_p$. Since $hit(p,\hat{r}_q)\geq L_h$ and
$hit(q,\hat{r}_q) < L_h$, we have $h_p > h_q$. It follows that if
index interval $[\pi_q,q]$ is an endorsed interval, then index
interval $[\pi_q,p]$ is an endorsed interval and
$conf(\pi_q,p)>conf(\pi_q,q).$ Thus, index interval $[\pi_q,q]$ is
impossible to be a maximum-confidence endorsed interval.
\end{proof}

\subsection {Subroutine}
We next give a subroutine to compute $r_q$ for all good indices $q$
in $O(n)$ time. The pseudocode is given below, where the goal is to
fill in an array $R[1..n]$, initialized with -1's, such that
$R[i]=r_i$ for all good indices $i$ and $R[i]=-1$ for all bad
indices $i$ at the end.

\begin{description}
\item[Subroutine]\textsc{RMP}\\
\end{description}
\vspace{-34pt}
    \begin{algorithmic}[1]\itemsep = -3pt
         \STATE create an array $R[1..n]$ initialized with -1's;
         \STATE $\hat{R}\leftarrow 0$;
                \STATE create an empty list $C$;
               \FOR{$i\leftarrow$ 1 to $n$}
                \WHILE{$C$ is not empty and $hit(C.lastElement()+1,i)\leq 0$}
                       \STATE delete from $C$ its last element;
                \ENDWHILE
                    \STATE insert $i$ at the end of $C$;

                \IF{$hit(\hat{R},i)\geq L_h$ or $hit(C.firstElement(),i)\geq L_h$}
                   \WHILE{$C$ is not empty and $hit(C.firstElement(),i)\geq L_h$}
                    \STATE $\hat{R}\leftarrow C.firstElement();$
                    \STATE delete from $C$ its first element;
                   \ENDWHILE
                   \STATE $R[i]\leftarrow \hat{R};$
                \ENDIF
                \ENDFOR
                \STATE output $R[1..n]$;
                \end{algorithmic}

\vspace{10pt} The Subroutine \textsc{RMP} consists of $n$
iterations, in the $i^{th}$ iteration, $R[i]$ is reset to $r_i$ if
$i$ is an good index. To accomplish this task efficiently, we
maintain a list $C$ and a variable $\hat{R}$ such that at the end of
the $i^{th}$ iteration the following conditions hold. \vspace{-2pt}
\begin{enumerate}\itemsep=-3pt
\item[1.] For each two adjacent elements $p<q$ in $C$, $hit(p+1,q)>0.$
\item[2.] For any index $q\geq i$, if $r_q\in[\hat{R},i]$, then $r_q\in C\cup\{\hat{R}\}.$
\item[3.] $\hat{R}=\hat{r}_i$.
\end{enumerate}
\vspace{-2pt} It is clear that $R[1]=r_1=0$ and the three conditions
hold at the end of the first iteration of the for-loop. Suppose that
the three conditions hold at the end of the $(i-1)^{th}$ iteration
of the for-loop. We shall prove that $R[i]$ will be reset to $r_i$
in the $i^{th}$ iteration of the for-loop if and only if $i$ is a
good index and the three conditions hold at the end of the $i^{th}$
iteration of the for-loop. Consider the moment immediately before
the execution of line $9$ in the $i^{th}$ iteration of the for-loop.
It is clear that condition~1 still holds at this moment. We next
prove that condition~2 also holds at this moment. Suppose for
contradiction that some $r_q$, where $q\geq i$, is removed from $C$
during the execution of the while-loop of lines $4\sim7.$ A
necessary condition for $r_q$ to be deleted is $hit(r_q+1,i)\leq 0$.
It follows that $hit(i,q)\geq hit(r_q,q)\geq L_h$, so $r_q$ is not
the right most partner of $q$, a contradiction. Therefore, we have
conditions~1~and~2 hold and $\hat{R}=\hat{r}_{i-1}$ just before the
execution of line 9. We next examine the execution of lines
$9\sim15$. Consider the following three cases.

Case~1: $r_i\not\in C\cup\{\hat{R}\}$ just before the execution of line 9. By condition~2, we have $r_i<\hat{r}_{i-1}$, so index $i$ is not good and  $\hat{r}_{i-1}=\hat{r}_i$. Thus, condition~3 holds just before the execution of line~9. In this case we will fail the test condition in line~9 so lines$~10\sim14$ will not be executed. Therefore, $R[i]$ is not reset, and conditions~$1\sim3$ hold at the end of the $i^{th}$ iteration of the for-loop.

Case~2: $r_i=\hat{R}$ just before the execution of line~9. It follows that index $i$ is good and $\hat{r}_{i-1} = r_{i} = \hat{r}_{i}$. Thus, condition~3 holds just before the execution of line~9. The body of the while-loop of lines $10\sim13$ will not be executed in this case. Therefore, $R[i]$ is reset to $\hat{R}=r_i$ in line 14, and conditions$~1\sim3$ hold at the end of the $i^{th}$ iteration of the for-loop.

Case~3: $r_i\in C$ just before the execution of line~9. It follows that $r_i>\hat{r}_{i-1}$, so $\hat{r}_i=r_i$ and index $i$ is good.
By conditions~1, for all $c\in C$ before $r_i$, $hit(c,i)> hit (r_i,i)\geq L_h$. Moreover, since $r_i$ is the right most partner of $i$ and indices in $C$ are in increasing order, for all $c'\in C$ after $r_i$, $hit(c',i)<L_h$. Therefore, we have $\hat{R}=r_i$ holds after the execution of the while-loop of lines $10\sim13$. It follows that condition~3 holds at the end of the $i^{th}$ iteration of the for-loop since $r_i = \hat{r}_i$. It is clear that deleting from $C$ a prefix has no harm to condition~1, so it remains to prove that condition~2 still holds. Suppose for contradiction that condition~2 does not hold at the end of the $i^{th}$ iteration of the for-loop. It follows that some $r_q$ with $r_i<r_q\leq i$ is removed from $C,$ which leads to a contradiction because $r_i$ is the largest index removed from $C$ in the while-loop of lines $10\sim13$.

We next analyze the running time. In each iteration of the while-loop of lines $5\sim7$ and the while-loop of lines $10\sim13$, there is one index in $C$ removed. Since each index is inserted into $C$ at most once, the time spent on these two while-loops is bounded by $O(n)$. To summarize, we have the following lemma.

\begin{lemma}\label{LEMMA: RMP}
Subroutine \textsc{RMP} computes in $O(n)$ time an array $R[1..n]$
such that $R[i]$ is the right most partner of $i$ if $i$ is a good
index and $R[i]=-1$ if $i$ is a bad index.
\end{lemma}

Denote by $\phi(x,y)$ the largest $z$ in $[x,y]$ that minimizes
$conf(x,z)$. We next describe a subroutine which finds the largest
$p\in [l,u]$ that maximizes $conf(p,q)$ given $0\leq l\leq u \leq q$. The
pseudocode is given below, where we initialize a variable $p$ with
value $l$ and then repeatedly resetting $p$ to $\phi(p,u-1)+1$ until
$p=u$ or the lowest confidence interval starting from $p$ and ending
before $u$ has the same confidence as interval $[p,q]$. The
pseudocode is given below.

\begin{description}
\item[Subroutine]BEST$(l,u,q)$\\
\end{description}
\vspace{-34pt}
    \begin{algorithmic}[1]\itemsep = -3pt
            \STATE $p\leftarrow l;$
            \WHILE{$p<u$ and $conf(p,\phi(p,u-1))\leq conf(p,q)$}
              \STATE $p\leftarrow \phi(p,u-1)+1$;
            \ENDWHILE
              \STATE output $p$;
    \end{algorithmic}

\begin{lemma}\label{LEMMA: BEST}~\cite{Chung}
The call to BEST$(l,u,q)$ will return the largest $p\in [l,u]$ that
maximizes $conf(p,q)$ if $0\leq l \leq u \leq q$.
\end{lemma}

\begin{lemma}\label{LEMMA: BEST2}
Let $p$ be the return value of the call to BEST$(l,r_q,q)$. Then $p=\pi_q$ if
$\pi_q\in[l,r_q]$ and $0\leq l \leq r_q$.
\end{lemma}
\begin{proof}
Suppose that $p$ is not a partner of $q$, i.e., $hit(p,q)<L_h\leq
hit(r_q,q)$. It follows that $conf(p,q)<conf(r_q,q)$, which
contradicts Lemma~\ref{LEMMA: BEST}. Thus, $p$ must be a partner
of $q$. Suppose for contradiction that $p\neq\pi_q$. Since $p$ is a
partner of $q$, by the definition of best partners, we have either
($conf(p,q)<conf(\pi_q,q)$) or ($conf(p,q)=conf(\pi_q,q)$ and
$p<\pi_q$), which contradicts Lemma~\ref{LEMMA: BEST}.
\end{proof}

Chung and Lu~\cite{Chung} also gave efficient implementations of Subroutine BEST in their paper, which directly implies the next lemma.

\begin{lemma}\label{LEMMA: IMPLEMENTATION}~\cite{Chung}
   A sequence of consecutive calls to Subroutine BEST, say BEST$(l_1,u_1,q_1)$,
   BEST$(l_2,u_2,q_2)$,$\ldots$, BEST$(l_k,u_k,q_k)$,
   can be completed in total $O(u_k+k)$ time
   provided that $l_1=0$, $l_i =$ BEST$(l_{i-1},u_{i-1},q_{i-1})$ for each $i = 2,3,\ldots, k$, $u_1\leq u_2 \leq\cdots \leq u_k$, and $u_i \leq q_i$ for each $i = 1,2,\ldots, k$.
\end{lemma}
\newpage
\subsection{Algorithm}
Our algorithm for the HCI problem is as follows. First, we
initialize a variable $l$ with value 0 and call Subroutine
\textsc{RMP} to compute an array $R[1..n]$ such that $R[i]=r_i$ if
$i$ is a good index and $R[i]=-1$ if $i$ is a bad index. Then, for
each good index $q$, taken in increasing order, call
BEST$(l,R[q],q)$ to compute the largest $l_q\in[l,R[q]]$ that
maximizes $conf(l_q,q)$ and reset variable $l$ to $l_q$. Finally,
the index interval $[l_q,q]$ that maximizes $conf(l_{q},q)$ is
returned. The pseudocode is given below.

\begin{description}
\item[Algorithm]\textsc{ComputeHCI}\itemsep=-23pt\\
\item[Input:] A sequence $D=((h_1,s_1), (h_2,s_2), \ldots, (h_n,s_n))$ of number pairs, where $s_i>0$ for all $i$,
    and a hit-support lower bound $L_h$.\\
\item[Output:] An index interval $I=[i,j]$ maximizing $conf(i,j)$ subject to $hit(i,j)\geq L_h$.\\
\vspace{-26pt}
\end{description}
  \begin{algorithmic}[1]
           \STATE $l\leftarrow 0$;
           \STATE $R\leftarrow$ call Subroutine RMP;
           \STATE $c_{max}\leftarrow -\infty$;
           \STATE $(\alpha,\beta) = (0,0)$;
           \FOR{$q\leftarrow 1$ to $n$}
                \IF{$R[q]\neq -1$}
                    \STATE $l\leftarrow$ BEST$(l,R[q],q)$;
                     \IF{$conf(l,q)>c_{max}$}
                     \STATE $c_{max} \leftarrow conf(l,q);$
                     \STATE $(\alpha,\beta)\leftarrow (l,q)$;
                    \ENDIF
                \ENDIF
           \ENDFOR
           \STATE output $[\alpha,\beta]$;
    \end{algorithmic}

\begin{theorem}
 Algorithm \textsc{ComputeHCI} solves the HCI problem in $O(n)$ time.
\end{theorem}
\begin{proof}
We first prove the correctness. Let $Q=\{q_1,q_2,\ldots,q_k\}$ be
the set of good indices, where $q_1<q_2\cdots<q_k$ and $k=|Q|$. Let
$l_{q_0}=0$ and $l_{q_i}$ be the return value of the call to
Subroutine BEST in the $q_i^{th}$ iteration of the for-loop,
$i=1,2,\ldots,k$. Note that $l_{q_i}$ is the largest integer in
$[l_{q_{i-1}},r_{q_i}]$ that maximizes $conf(l_{q_i},q)$ for all $i$
with $1\leq i\leq k$. Let $q_{i^*}$ be the good index that maximizes
$conf(\pi_{q_{i^*}},q_{i^*})$. By Lemma~\ref{real < ideal}, it
suffices to prove that $\pi_{q_{i^*}}=l_{q_{i^*}}.$ To prove
$\pi_{q_{i^*}}=l_{q_{i^*}}$, by Lemma~\ref{LEMMA: BEST2}, it
suffices to prove that $\pi_{q_{i^*}}\in
[l_{q_{i^*-1}},r_{q_{i^*}}],$ i.e., $\pi_{q_{i^*}}\geq
l_{q_{i^*-1}}$. Suppose for contradiction that $\pi_{q_{i^*}}<
l_{q_{i^*-1}}$. Let $q_t\leq q_{i^*-1}$ be the first good index such
that $\pi_{q_{i^*}}< l_{q_{t}}.$ Then we have $\pi_{q_{i^*}}\in
[l_{q_{t-1}},r_{q_{t}}]$ and by Lemma~\ref{LEMMA: BEST}, $l_{q_t}$
is the largest index in $[l_{q_{t-1}},r_{q_t}]$ that maximizes
$conf(l_{q_t},q)$. It follows that
\begin{equation}conf(\pi_{q_{i^*}}, l_{q_{t}}-1)\leq conf(l_{q_{t}},
q_t).\end{equation} Since $conf(l_{q_{t}},q_t)\geq
conf(r_{q_{t}},q_t)\geq \frac{L_h}{sup(r_{q_{t}},q_t)}\geq 0$ and
$sup(l_{q_{t}},q_t)\geq sup(r_{q_{t}},q_t)$, we have
$hit(l_{q_{t}},q_t)\geq hit(r_{q_{t}},q_t)\geq L_h$. Therefore,
$l_{q_{t}}$ is a partner of $q_t$ and we have
\begin{equation}
conf(l_{q_{t}}, q_t)\leq conf(\pi_{q_{t}},q_{t}).\end{equation}
Following from inequality (1), we have $conf(l_{q_{t}}, q_t)\leq
conf(q_{t}+1,q_{i^*})$ for otherwise $conf(\pi_{q_{t}},q_{t})\leq
conf(\pi_{q_{i^*}},q_{i^*})< conf(l_{q_{t}}, q_t)$, which
contradicts inequality~(2). Combining inequality~(1) with
$conf(l_{q_{t}}, q_t)\leq conf(q_{t}+1,q_{i^*})$, we have
$conf(\pi_{q_{i^*}}, l_{q_{t}}-1)\leq conf(l_{q_{t}}, q_t)\leq
conf(q_{t}+1,q_{i^*}).$ It follows that
\begin{equation}conf(\pi_{q_{i^*}},q_{i^*})\leq
conf(l_{q_{t}},q_{i^*}).\end{equation} By inequality~(3),
$0\leq\frac{L_h}{sup(r_{q_{i^*}},q_{i^*})}\leq
conf(r_{q_{i^*}},q_{i^*})$, and $conf(r_{q_{i^*}},q_{i^*})\leq
conf(\pi_{q_{i^*}},q_{i^*})$, we have
\begin{equation}
0\leq conf(r_{q_{i^*}},q_{i^*})\leq conf(l_{q_{t}},q_{i^*}).
\end{equation}
Since $l_{q_t}\leq r_{q_t} \leq r_{q_{i^*}}\leq q_{i^*} $, we have
$sup(r_{q_{i^*}},q_{i^*})\leq sup(l_{q_t},q_{i^*}).$ By
inequality~(4) and $sup(r_{q_{i^*}},q_{i^*})\leq
sup(l_{q_t},q_{i^*})$, we have
\begin{equation}
L_h\leq hit(r_{q_{i^*}},q_{i^*})\leq hit(l_{q_{t}},q_{i^*}).
\end{equation}
By inequalities~(3) and (5), $l_{q_t}$ is a partner of $q_{i^*}$ and
$conf(\pi_{q_{i^*}},q_{i^*})\leq conf(l_{q_{t}},q_{i^*})$, which
contradicts the definition of best partners.

We now analyze the time complexity. By Lemma~\ref{LEMMA: RMP}, the
call to \textsc{RMP} takes $O(n)$ time. By the definition of good
indices, we have $R[q_1]\leq R[q_2]\leq \cdots \leq R[q_k]$. Because
we also have $l_{q_i} =$ BEST$(l_{q_{i-1}},R[q_i],q_i)$ for each
$i=1,2,\ldots,k$, by Lemma~\ref{LEMMA: IMPLEMENTATION}, the calls to
Subroutine~BEST, i.e., BEST$(l_{q_0},R[q_1],q_1)$,
BEST$(l_{q_{1}},R[q_2],q_2)$,$\ldots,$ and
BEST$(l_{q_{k-1}},R[q_k],q_k)$, totally take $O(R[q_k]+k)=O(n)$
time.
\end{proof}

Finally, we modify Algorithm ComputeHCR such that it runs in an
online manner. The modified version is given below, where we
maintain an integer pair $(\alpha,\beta)$ such that after processing
the $q^{th}$ number pair in $D$ at the $q^{th}$ iteration of the
for-loop, the integer interval $[\alpha,\beta]$ will be a
maximum-confidence endorsed interval for the subsequence
$((h_1,s_1), (h_2,s_2), \ldots, (h_q,s_q)),$ if any. The correctness
is easy to verify by noting that $r=R[q]$ holds at the end of the
$q^{th}$ iteration of the for-loop for each $q=1,2,\ldots,n$.

\begin{theorem}
 Algorithm \textsc{OnlineComputeHCR} solves the HCI problem in $O(n)$ time in an online manner.
\end{theorem}
\newpage

\begin{description}
\item[Algorithm]\textsc{OnlineComputeHCI}\itemsep=-23pt\\
\item[Input:] A sequence $D=((h_1,s_1), (h_2,s_2), \ldots, (h_n,s_n))$ of number pairs, where $s_i>0$ for all $i$,
    and a hit-support lower bound $L_h$.\\
\item[Goal:] Maintaining an integer pair
$(\alpha,\beta)$ such that after processing the $q^{th}$ number pair
in $D$ at the $q^{th}$ iteration of the for-loop, the integer
interval $[\alpha,\beta]$
is a maximum-confidence endorsed interval for the subsequence $((h_1,s_1), (h_2,s_2), \ldots, (h_q,s_q))$, if any.\\
\vspace{-28pt}
\end{description}
              \begin{algorithmic}[1]
                \STATE $l\leftarrow 0$;
                \STATE $r\leftarrow -1$;
                \STATE $\hat{r}\leftarrow 0$;
                \STATE $c_{max}\leftarrow -\infty$;
                \STATE $(\alpha,\beta)\leftarrow (0,0)$;
                \STATE create an empty list $C$;
                \FOR{$q\leftarrow 1$ to $n$}
                    \STATE $r\leftarrow -1$;
                    \WHILE{$C$ is not empty and $hit(C.lastElement()+1,q)\leq 0$}
                       \STATE delete from $C$ its last element;
                    \ENDWHILE
                \STATE insert $q$ at the end of $C$;
                \IF{$hit(\hat{r},q)\geq L_h$ or $hit(C.firstElement(),q)\geq L_h$}
                   \WHILE{$C$ is not empty and $hit(C.firstElement(),q)\geq L_h$}
                    \STATE $\hat{r}\leftarrow C.firstElement();$
                    \STATE delete from $C$ its first element;
                   \ENDWHILE
                   \STATE $r\leftarrow \hat{r};$
                \ENDIF
                \IF{$r\neq -1$}
                    \STATE $l\leftarrow$ BEST$(l,r,q)$;
                    \IF{$conf(l,q)>c_{max}$}
                    \STATE $c_{max} \leftarrow conf(l,q);$
                    \STATE $(\alpha,\beta)\leftarrow (l,q)$;
                    \ENDIF
                \ENDIF
           \ENDFOR
            \end{algorithmic}

\newpage

\section {A subqurdratic time algorithm for the PSEI problem}
Define the length of an index interval $[i,j]$ to be $length(i,j)=j-i+1$.
In the PSEI problem, the input sequence $D$ is plain, so we have
$length(i,j)=sup(i,j)=j-i+1$ and $ecc(i,j)
=\frac{hit(i,j)}{\sqrt{length(i,j)}}$. Thus, we can reformulate the
PSEI problem as follows: Given a sequence $D=((h_1,s_1), (h_2,s_2),
\ldots, (h_n,s_n))$ of number pairs, where $s_i=1$ for all $i$, and
a length lower bound $L=L_s$, find an index interval $I=[i,j]$ maximizing
the eccentricity $ecc(i,j)= \frac{hit(i,j)}{\sqrt{length(i,j)}}$
subject to $length(i,j)\geq L$.

\subsection {Preliminaries}
Let $H=(h_1,h_2,\ldots,h_n)$ and $P_H[0..n]$ be the prefix-sum array
of $H$, where $P_H[0]=0$ and $P_H[i]=P_H[i-1]+h_i$ for each
$i=1,2,\ldots,n$. Note that $hit(i,j) = P_H[i]-P_H[j-1]$,
$length(i,j) = j-i+1$ and $ecc(i,j) =
\frac{P_H[i]-P_H[j-1]}{\sqrt{j-i+1}}$.  Thus, after constructing
$P_H$ in $O(n)$ time, each computation of the hit-support, length,
or eccentricity of an index interval can be done in constant time.
In the following, we review some definitions and theorems. For more
details, readers can refer to
\cite{Aho,Bernholt2007,Bernholt,Lin,Takaoka}.

\begin{definition}
A function $f:\mathbb{R}^+\times\mathbb{R}\rightarrow \mathbb{R}$ is
said to be \emph{quasiconvex} if and only if for all points $u,v\in
\mathbb{R}^+\times\mathbb{R}$ and all $\lambda \in [0,1]$, we have
$f(\lambda\cdot u + (1-\lambda)\cdot v) \leq \max\{f(u),f(v)\}$.
\end{definition}

\begin{lemma}\cite{Bernholt}\label{quasiconvex}
Define $f:\mathbb{R}^+\times\mathbb{R}$ by letting
\[ f(\ell,h)=
    \left\{
        \begin{array}{ll}
             \frac{h}{\sqrt \ell} & \mbox{if } h\geq0;\\
             0 & \mbox{otherwise}.
        \end{array}
    \right.
    \]
Then $f$ is quasiconvex.
\end{lemma}

\begin{theorem}\cite{Bernholt2007}\label{MSI}
Given a sequence of $n$ number pairs
$D=((h_1,s_1),(h_2,s_2),\ldots,(h_n,s_n))$, a length lower bound
$L$, and a quasiconvex score function
$f:\mathbb{R}^+\times\mathbb{R}\rightarrow \mathbb{R}$, there exists
an algorithm, denoted by MSI$(D,L,f)$, which can find an index interval
$[i,j]$ that maximizes $f(length(i,j),hit(i,j))$ subject to
$length(i,j)\geq L$ in $O(n)$ time.
\end{theorem}

By the fact $f(\ell,h)=h$ is quasiconvex and Theorem~\ref{MSI}, we
have the following corollary, which was also proved
in~\cite{Bernholt2007,Fan,Lin}.

\begin{corollary}\label{max sum}
There exists an $O(n)$-time algorithm for finding an index interval
$[i,j]$ maximizing $hit(i,j)$ subject to $length(i,j)\geq L$.
\end{corollary}
\newpage
We next prove that if all index intervals with lengths $\geq L$ have
negative hit-supports, then the optimal solution must have length
less than~$2L$.

\begin{lemma}\label{2L}
If $hit(p,q)<0$ holds for each index interval $[p,q]$ of length at least $L$,
then $length(p^*,q^*)<2L$, where $(p^*,q^*)= \arg
\max\limits_{length(p,q)\geq L} ecc(p,q)$.
\end{lemma}
\begin{proof}
Let $(p^*,q^*)= \arg \max\limits_{length(p,q)\geq L} ecc(p,q)$.
Suppose for contradiction that $length(p^*,q^*)\geq 2L$. Let
$c_1=\lfloor (p^*+q^*)/2 \rfloor$ and $c_2=c_1+1$. Then we have
$length(p^*,c_1)\geq L$ and $length(c_2,q^*)\geq L$. Since
$\frac{hit(p^*,q^*)}{length(p^*,q^*)} =
\frac{hit(p^*,c_1)+hit(c_2,q^*)}{length(p^*,c_1)+length(c_2,q^*)} ,$
we have \[\frac{hit(p^*,q^*)}{length(p^*,q^*)}\leq
\frac{hit(p^*,c_1)}{length(p^*,c_1)} \mbox{ or }
\frac{hit(p^*,q^*)}{length(p^*,q^*)}\leq
\frac{hit(c_2,q^*)}{length(c_2,q^*)}.\] Without loss of generality,
we assume $\frac{hit(p^*,q^*)}{length(p^*,q^*)}\leq
\frac{hit(p^*,c_1)}{length(p^*,c_1)}$. Since
$\frac{hit(p^*,q^*)}{length(p^*,q^*)}\leq
\frac{hit(p^*,c_1)}{length(p^*,c_1)}<0$ and
$\sqrt{length(p^*,q^*)}>\sqrt{length(p^*,c_1)}$, we have

\begin{eqnarray*}
&&\frac{\sqrt{length(p^*,q^*)}\cdot hit(p^*,q^*)}{length(p^*,q^*)}<
\frac{\sqrt{length(p^*,c_1)}\cdot
hit(p^*,c_1)}{length(p^*,c_1)}\\
&\Leftrightarrow& \frac{hit(p^*,q^*)}{\sqrt{length(p^*,q^*)}}<
\frac{hit(p^*,c_1)}{\sqrt{length(p^*,c_1)}}\\
&\Leftrightarrow& ecc(p*,q*)<ecc(p^*,c_1).
\end{eqnarray*}
It contradicts $(p^*,q^*)= \arg \max\limits_{length(p,q)\geq
L}ecc(p,q)$.
\end{proof}

\subsection{Subroutine}
In the following we give a new algorithm for the \textsc{Min-Plus
Convolution Problem}, which will serve as a subroutine in our
algorithm for the PSEI problem. The min-plus convolution of two
vectors \textbf{x}~$=(x_0,x_1,$ $\ldots,x_{n-1})$ and
\textbf{y}~$=(y_0,y_1,\ldots,y_{n-1})$ is a vector
\textbf{z}~$=(z_0,z_1,\ldots,z_{n-1})$ such that
$z_k=\min_{i=0}^{k}\{x_i+y_{k-i}\}$ for $k=0,1,\ldots,n-1$. Given
two vectors \textbf{x}~$=(x_0,x_1,\ldots,x_{n-1})$ and
\textbf{y}~$=(y_0,y_1,\ldots,y_{n-1})$, the \textsc{Min-Plus
Convolution Problem} is to compute the min-plus convolution
\textbf{z} of \textbf{x} and \textbf{y}. This problem has appeared
in the literature with various names such as ``minimum
convolution,'' ``epigraphical sum,'' ``inf-convolution,'' and
``lowest midpoint''
\cite{Bellman,Bergkvist,Felzenszwalb,Maragos,Moreau,Rockafellar,Stromberg}.
Although it is easy to obtain an $O(n^2)$-time algorithm, no
subquadratic algorithm was known until recently Bremner $et~al$.
\cite{Bremner} proposed an $O(n^2/\log n)$-time algorithm. In the
following, we shall give an $O(n^{1/2}T(n^{1/2}))$-time algorithm
for the \textsc{Min-Plus Convolution Problem}, where $T(n)$ is the
time required to solve the all-pairs shortest paths problem on a
graph of $n$ nodes. To date, the best algorithm for computing the
all-pairs shortest paths problem on a graph of $n$ nodes runs in
$O(n^3\frac{(\log\log n)^3}{(\log n)^2})$ time \cite{Chan}. Thus,
our work implies an $O(n^2\frac{(\log\log n)^3}{(\log n)^2})$-time
algorithm for the min-plus convolution problem, which is slightly
superior to the first subquadratic $O(n^2/\log n)$-time algorithm
recently proposed by Bremner $et~al$. \cite{Bremner}.

\begin{definition}
The min-plus product $BC$ of a $d\times n'$ matrix $B=[b_{i,j}]$ and
an $n'\times d$ matrix $C=[c_{i,j}]$ is a $d\times d$ matrix
$D=[d_{i,j}]$ where $d_{i,j}=\min_{k=0}^{n'-1}\{b_{i,k}+c_{k,j}\}$.
\end{definition}

Note that the notion of ``min-plus product'' is different from the
notion of ``min-plus convolution''. It is well known \cite{Aho} that
the time complexity of computing the min-plus product of two
$n'\times n'$ matrices is asymptotically equal to that of computing
all pairs shortest paths for a graph with $n'$ vertices. The next
lemma was proved by Takaoka in~\cite{Takaoka}. The proof of the next
lemma was also given in \cite{Takaoka}, and we include it here for
completeness.

\begin{lemma}\label{min-plus product}~\cite{Takaoka}
Given a $T(n')$-time algorithm for computing the min-plus product of
any two $n'\times n'$ matrices, the computation of the min-plus
product of $B$ and $C$, where $B$ is a $d\times n'$ matrix and $C$
is an $n'\times d$ matrix, can be done in $O(\frac{n'}{d}T(d))$ time
if $d\leq n'$.
\end{lemma}
\begin{proof}
For simplicity we assume that $d$ divides $n$. We first split $B$
into $n'/d$ matrices $B_1,\ldots,B_{n'/d}$ of dimension $d\times d$
and $C$ into $n'/d$ matrices $C_1,\ldots,C_{n'/d}$ of dimension
$d\times d$. Then we can compute $\{B_1C_1,B_2C_2,$
$\ldots,B_{n'/d}C_{n'/d}\}$ in $O(dT(n'/d))$ time by the given
algorithm. The $(i,j)$-th entry of the min-plus product of $B$ and
$C$ is $\min_{k=1}^{n'/d}$\{the $(i,j)$-th entry of $B_kC_k$\}.
\end{proof}

Our new algorithm for the \textsc{Min-Plus Convolution Problem} is
as follows.
\vspace{-10pt}
\begin{tabbing}
\hspace*{2em} \= \hspace*{2em} \= \hspace*{2em} \= \hspace*{2em} \= \hspace*{2em} \kill \\
\textbf{Algorithm} {\sc MinPlusConvolution} \\
\textbf{Input:} \textbf{x}~$=(x_0,x_1,\ldots,x_{n-1})$ and
\textbf{y}~$=(y_0,y_1,\ldots,y_{n-1})$.\\
\textbf{Output:} \textbf{z}~$=(z_0,z_1,\ldots,z_{n-1})$ such that
$z_k=\min_{i=0}^{k}\{x_i+y_{k-i}\}$ for $k=0,1,\ldots,n-1$.\\
\  1: Construct an $\lceil n^{1/2} \rceil \times (2n-1)$ matrix
$B=[b_{i,j}]$ such that the $i^{th}$ row of $B$ \\
\ \> is equal to
$(\infty,\ldots,\infty,x_0,x_1,\ldots,x_{n-1},\overbrace{\infty,\ldots,\infty}^{i
\times \lceil n^{1/2} \rceil})$ for $i=0,1,\ldots,\lceil n^{1/2} \rceil -1$.\\
\  2: Construct a $(2n-1) \times \lceil n^{1/2} \rceil$ matrix $C=[c_{i,j}]$ such that the transpose \\
\>\ of the $j^{th}$ column of $C$ is equal to $(\overbrace{\infty,\ldots,\infty}^{j},y_{n-1},y_{n-2},\ldots,y_{0},\infty,\ldots,\infty)$\\
\>\>\  for $j=0,1,\ldots,\lceil n^{1/2} \rceil -1$.\\
\  3: Let $D~=~[d_{i,j}]$ be the min-plus product of $B$ and $C$.\\
\  4: For $k=0,1,\ldots,n-1$ do\\
\> Find $i,j$ such that $k=i \times \lceil n^{1/2} \rceil +j$, where $0\leq j<\lceil n^{1/2} \rceil$.\\
\>Set $z_k$ to $d_{i,j}$.\\
\  5: Output \textbf{z}~$ = (z_0,z_1,\ldots,z_{n-1})$.\\
\end{tabbing}

The following lemma ensures the correctness.

\begin{lemma}
In {\sc MinPlusConvolution}, $d_{i,j}=\min_{t=0}^{i \times \lceil
n^{1/2} \rceil +j}\{x_t+y_{i \times \lceil n^{1/2} \rceil +j-t}\}$
if $0\leq i \times \lceil n^{1/2} \rceil +j\leq n-1$.
\end{lemma}
\begin{proof}
\begin{eqnarray*}
d_{i,j}& = & \min_{t=0}^{2n-1}\{b_{i,t}+c_{t,j}\}\\
 &=& \min \{
\min_{t=0}^{n-2-i \times \lceil n^{1/2}
\rceil}\{b_{i,t}+c_{t,j}\},\min_{t=n-1-i \times \lceil
n^{1/2}\rceil}^{n+j-1}\{b_{i,t}+c_{t,j}\},
\min_{t=n+j}^{2n-1}\{b_{i,t}+c_{t,j}\} \}\\ &=&  \min \{
\min_{t=0}^{n-2-i \times \lceil n^{1/2}
\rceil}\{\infty+c_{t,j}\},\min_{t=n-1-i \times \lceil
n^{1/2}\rceil}^{n+j-1}\{b_{i,t}+c_{t,j}\},
\min_{t=n+j}^{2n-1}\{b_{i,t}+\infty\} \}\\ &=& \min_{t=n-1-i \times
\lceil n^{1/2}\rceil}^{n+j-1}\{b_{i,t}+c_{t,j}\}\\ &=&
\min\{x_{0}+y_{i \times \lceil n^{1/2} \rceil +j}, x_{1}+y_{i \times
\lceil n^{1/2} \rceil +j-1}+\cdots+ x_{i \times \lceil n^{1/2}
\rceil +j}+y_{0} \}\\ &=& \min_{t=0}^{i \times \lceil n^{1/2} \rceil
+j}\{x_t+y_{i \times \lceil n^{1/2} \rceil +j-t}\}
\end{eqnarray*}
\end{proof}

We now analyze the time complexity. Let $T(n)$ denote the time
required to compute the min-plus product of two $n\times n$
matrices. Steps~1 and 2 take $O(n^{3/2})$ time, and by
Lemma~\ref{min-plus product}, Step~3 takes
$O(\frac{2n-1}{n^{1/2}}T(\lceil
n^{1/2}\rceil))=O(n^{1/2}T(n^{1/2}))$ time. Steps 4 and 5 take
$O(n)$ time. Therefore, the total running time is
$O(n^{1/2}T(n^{1/2})+n^{3/2})$. Since $T(n)=\Omega(n^2)$, we have
$O(n^{1/2}T(n^{1/2})+n^{3/2})=O(n^{1/2}T(n^{1/2}))$.
Theorem~\ref{min-plus convolution} summarizes our results for the
\textsc{Min-Plus Convolution Problem}.

\begin{theorem}\label{min-plus convolution}
The running time of Algorithm {\sc MinPlusConvolution} is
$O(n^{1/2}T(n^{1/2}))$, where $T(n)$ is the time required to compute
the min-plus product of two $n\times n$ matrices.
\end{theorem}

The next Lemma was proved by Bergkvist and Damaschke
in~\cite{Bergkvist}.

\begin{lemma}\cite{Bergkvist}\label{reduction}
Given a sequence $H=(h_1,h_2,\ldots,h_n)$, the \textsc{Maximum
Consecutive Sums Problem} is to compute a sequence
$(w_1,w_2,\ldots,w_n)$ where $w_i=\max\{\sum_{j=p}^q h_j:
length(p,q) = i\}$ for each $i=1,2,\ldots,n$. The \textsc{Maximum
Consecutive Sums Problem} can be reduced to the \textsc{Min-Plus
Convolution Problem} in linear time.
\end{lemma}

\begin{corollary}\label{MaxCW}
Given a sequence $H=(h_1,h_2,\ldots,h_n)$, we can compute in
$O(n^{1/2}T(n^{1/2}))$ time a sequence $(w_1,w_2,\ldots,w_n)$ such
that $w_i=\max\{\sum_{j=p}^q h_j: length(p,q) = i\}$ for each
$i=1,2,\ldots,n$ by making use of the Alogirthm {\sc
MinPlusConvolution}, where $T(n)$ is the time required to compute
the min-plus product of two $n\times n$ matrices.
\end{corollary}
\begin{proof}
Immediately from Theorem~\ref{min-plus convolution} and
Lemma~\ref{reduction}.
\end{proof}

\subsection {Algorithm}
We next show how to solve the PSEI problem in
$O(n\frac{T(L^{1/2})}{L^{1/2}})$ time, where $T(n')$ is the time
required to compute the min-plus product of two $n'\times n'$
matrices. To avoid notational overload, we assume that $4L$ divides
$n$. By Corollary~\ref{max sum}, we can find an index interval $[i,j]$
maximizing $hit(i,j)$ subject to $length(i,j)\geq L$ in linear time.
Then there are three cases to consider: (1) $hit(i,j)=0$; (2)
$hit(i,j)>0$; (3) $hit(i,j)<0$. If it is Case~1, then $[i,j]$ must
be an optimal solution, and we are done. If it is Case~2, then we
know there is at least one index interval satisfying the length constraint
with positive hit-support. Define $f(\ell,h)$ by
\[ f(\ell,h)=
    \left\{
        \begin{array}{ll}
             \frac{h}{\sqrt \ell} & \mbox{if } h\geq0;\\
             0 & \mbox{otherwise},
        \end{array}
    \right.
    \]
By Lemma~\ref{quasiconvex}, $f$ is quasiconvex,  so we can call
MSI$(D,L,f)$ to find the index interval $[i',j']$ maximizing $f(i',j')$
subject to $length(i',j')\geq L$ in linear time. Clearly,  $[i',j']$
is an optimal solution, and we are done. If it is Case~3, i.e., all
index intervals satisfying the length constraint have negative
hit-supports, we do the following. First, by letting $I_k$ be the
index interval $[2kL+1, 2kL+4L]$ for each $k=0, 1,\ldots,\frac{n}{2L}-2$,
we can divide the whole index interval $[1,n]$ into $\frac{n}{2L}-1$
subintervals, each of length $4L$.
By making use of Corollary~\ref{MaxCW}, we are able to compute the
index interval $[i_k,j_k]\subseteq I_k$ maximizing the eccentricity
subject to the length constraint in $O(L^{1/2}T(L^{1/2}))$ time for
each $k = 0, 1,\ldots, \frac{n}{2L}-2$. According to Lemma~\ref{2L},
some $[i_k,j_k]$ must be the optimal solution. The detailed
algorithm is given below.

\begin{description}
\item[Algorithm]\textsc{ComputePSEI}\itemsep=-23pt\\
\item[Input:] A plain sequence $D$ of $n$ number pairs and a length lower bound $L$ with $1\leq L \leq n$.\\
\item[Output:] An index interval $I=[i,j]$ maximizing $ecc(i,j)$ subject to $length(i,j)\geq L$.\\
\vspace{-29pt}
\end{description}
    \begin{description}\itemsep=-21pt
        \item[\ \textmd{1:}] $(i,j)\leftarrow\arg \max\limits_{length(p,q)\geq L} hit(p,q)$.\\
        \item[\ \textmd{2:}] If $hit(i,j)=0$, then return $[i,j]$.\\
        \item[\ \textmd{3:}] If $hit(i,j)>0$ then\\
            \vspace{-27pt}
            \begin{enumerate}\itemsep=-21pt
                \item[\ \textmd{1:}] define $f:\mathbb{R}^+\times\mathbb{R}$ by letting
                $f(\ell,h)=\frac{h}{\sqrt{\ell}}$ if $h\geq0$ and 0, otherwise;\\
                \item[\ \textmd{2:}] return MSI$(D,L,f)$.\\
            \end{enumerate}
            \vspace{-7pt}
        \item[\ \textmd{4:}] For $k$ from  0 to $\frac{n}{2L}-2$,\\
        \vspace{-27pt}

              \begin{enumerate}\itemsep=-1pt
                \vspace{-2pt}
                \item compute $(w_1,\ldots,w_{4L})$, where
                $w_j=\max\{hit(p,q)|\ length(p,q)=j \mbox{ and } 2kL+1 \leq p\leq q \leq 2kL+4L\}$ for each $j=1,\ldots,4L$;
                \item compute $ecc_j = \frac{w_j}{\sqrt{j}}$ for each $j=L,L+1,\ldots,2L-1$;
                \item
                 $j_k\leftarrow \arg\max_{j=L}^{2L-1}ecc_j$;
                 \item
                 $i_k\leftarrow \arg\max_{i=2kL+1}^{2kL+4L-j_k+1}ecc(i,i+j_k-1)$.\\
              \end{enumerate}
               \vspace{-6pt}
           \item[\ \textmd{5:}] Return the max eccentricity interval in $\{[i_0,j_0],[i_1,j_1],\ldots,[i_{\frac{n}{2L}-2},j_{\frac{n}{2L}-2}]\}$.\\
\end{description}


\begin{theorem} \label{PSEI result}
Algorithm \textsc{ComputePSEI} solves the PSEI problem in
$O(n\frac{T(L^{1/2})}{L^{1/2}})$ time, where $T(n')$ is the time
required to compute the min-plus product of two $n'\times n'$
matrices.
\end{theorem}

\begin{proof}
We begin by considering the correctness. Let $(i,j)=\arg
\max\limits_{length(p,q)\geq L} hit(p,q)$ and $(i^*,j^*)=\arg
\max\limits_{length(p,q)\geq L} ecc(p,q)$.

In the case where $hit(i,j)=0$, we have
$ecc(p,q)=\frac{hit(p,q)}{\sqrt{length(p,q)}}\leq 0$ for all $(p,q)$
with $length(p,q)\geq L$. It follows that $0\geq ecc(i^*,j^*)\geq
ecc(i,j)= 0$, so index interval $[i,j]$ is an optimal solution.

In the case where $hit(i,j)>0$, we have $ecc(i^*,j^*)\geq
ecc(i,j)>0$. Let $f(\ell,h)=\frac{h}{\sqrt{\ell}}$ if $h\geq0$ and
0, otherwise. By Lemma~\ref{quasiconvex}, $f$ is quasiconvex, so the
call to MSI$(D,L,f)$ will return an index interval $[i',j']$ maximizing
$f(i',j')$ subject to $length(i',j')\geq L$. We next prove that
$ecc(i',j')\geq ecc(i^*,j^*)$, so index interval $[i',j']$ is an optimal
solution. Note that for any index interval $[p,q]$, we have $f(p,q) =
ecc(p,q)$ as long as $ecc(p,q)\geq 0$ or $f(p,q)> 0$. Since
$ecc(i^*,j^*) > 0$, we have $ecc(i^*,j^*) = f(i^*,j^*) >0$. It
follows that $f(i',j') \geq f(i^*,j^*) >0$, which implies that
$ecc(i',j')=f(i',j')$. Therefore, $ecc(i',j')=f(i',j')\geq
f(i^*,j^*) = ecc(i^*,j^*)$.

In the case where $hit(i,j)<0$, we have $length(i^*,j^*)<2L$ by
Lemma~\ref{2L}. It follows that $[i^*,j^*]$ must be contained in
$[2kL+1, 2kL+4L]$ for some $k\in \{0,1,\ldots,\frac{n}{2L}-2\}$ and
thus the index interval returned at Step~5 must be an optimal solution.

We next analyze the running time. By Corollary~\ref{max sum}, Step~1
takes $O(n)$ time. Step~2 takes constant time. By Theorem~\ref{MSI},
Step~3 takes $O(n)$ time. By Corollary~\ref{MaxCW}, each iteration
of the loop at Step~4 takes $O(L^{1/2}T(L^{1/2})+L)$ time. Thus
Step~4 takes
$O(\frac{n}{L}L^{1/2}T(L^{1/2})+n)=O(n\frac{T(L^{1/2})}{L^{1/2}})$
time. Step~5 takes $O(n/L)$ time. Therefore the total running time
is $O(n\frac{T(L^{1/2})}{L^{1/2}})$.
\end{proof}

\section{Concluding remarks}
To the best of our knowledge, there is not any non-trivial lower
bound for the PSEI problem proved so far. Thus, there is still a
large gap between the trivial lower bound of $O(n)$ and the upper
bound of $O(nL_s\frac{(\log\log L_s)^3}{(\log L_s)^2})$ for the PSEI
problem. Bridging this gap remains an open problem.

\section*{Acknowledgments}
   We thank Kuan-Yu Chen, Meng-Han Li, Cheng-Wei Luo, Hung-Lung Wang, and Roger Yang for helpful discussion.


\begin{thebibliography}{99}\small

     \bibitem{Aho}
    A. Aho, J. Hopcroft, and J. Ullman. \textit{The Design and Analysis
    of Computer Algorithms}, Addison-Wesley, Reading, MA, 1974.

     \bibitem {Allison03}
     L. Allison. Longest Biased Interval and Longest Non-Negative Sum Interval.
     \textit{Bioinformatics}, 19(10):1294-1295, 2003.

    \bibitem{Bellman}
    R. Bellman and W. Karush. Mathematical Programming and the Maximum Transform.
    \textit{Journal of the Society for Industrial and Applied
    Mathematics}, 10(3):550-567, 1962.

    \bibitem{Bergkvist}
    A. Bergkvist and P. Damaschke. Fast Algorithms for Finding Disjoint Subsequences with Extremal
    Densities. \textit{Pattern Recognition}, 39(12):2281-2292, 2006.

    \bibitem{Bernholt2007}
    T. Bernholt, F. Eisenbrand, and T. Hofmeister. A Geometric Framework for Solving Subsequence Problems in
    Computational Biology Efficiently. \textit{SoCG}, 310-318, 2007.


    \bibitem{Bernholt}
    T. Bernholt and T. Hofmeister. An Algorithm for a
    Generalized Maximum Subsequence Problem. \textit{LATIN}, 178-189, 2006.


    \bibitem{Bremner}
    D. Bremner, T. Chan, E. Demaine, J. Erickson, F. Hurtado, J. Iacono, S. Langerman, I. Streinu, and
    P.
    Taslakian. Necklaces, Convolutions, and X+Y. \textit{ESA}, 160-171, 2006.

    \bibitem{Chan}
    T. M. Chan. More Algorithms for All-Pairs Shortest Paths in
    Weighted Graphs. To appear in \textit{STOC}, 2007.

    \bibitem{Chen05}
    K.-Y. Chen and K.-M Chao. Optimal Algorithms for Locating the Longest and Shortest Segments
    Satisfying a Sum or an Average Constraint.
    \textit{Information Processing Letters}, 96(6):197-201, 2005.


    \bibitem{Cheng08}
    C.-H Cheng, H.-F. Liu, and K.-M. Chao. Optimal Algorithms for the Average-Constrained Maximum-Sum
    Segment Problem. \textit{Information Processing Letters}, accepted, 2008.

    \bibitem{Chung}
    K.-M. Chung and H.-I Lu. An Optimal Algorithm for the Maximum-Density Segment
    Problem. \textit{SIAM Journal on Computing}, 34(2):373-387, 2004.


    \bibitem{Davies}
    P. L. Davies and A. Kovac. Local Extremes, Runs,
    Strings and Multiresolution (with discussion). \textit{Annals
    of Statistics}, 29(1):1-65, 2001.

    \bibitem{Fan}
    T.-H. Fan, S. Lee, H.-I Lu,
    T.-S. Tsou, T.-C. Wang, and A. Yao. An Optimal Algorithm for
    Maximum-Sum Segment and Its Application in Bioinformatics Extended
    Abstract. \textit{CIAA}, 251-257, 2003.

    \bibitem{Felzenszwalb}
    P. Felzenszwalb and D. Huttenlocher. Distance Transforms of Sampled Functions.
    Technical Report TR2004-1963, Cornell Computing and Information Science, 2004.

      \bibitem{Fukuda99}
    T. Fukuda, Y. Morimoto, S. Morishita, and T. Tokuyama. Mining Optimized Association Rules for Numeric Attributes. \textit{Journal of Computer and System Science}, 58(1):1-12, 1999.

    \bibitem{Fukuda01}
    T. Fukuda, Y. Morimoto, S. Morishita, and T. Tokuyama. Data Mining with Optimized Two-Dimensional Association Rules. \textit{ACM Transactions on Database Systems}, 26(2):179-213, 2001.

    \bibitem{Goldwasser}
    M. Goldwasser, M.-Y. Kao, and H.-I Lu. Linear-Time Algorithms for
    Computing Maximum-Density Sequence Segments with Bioinformatics
    Applications. \textit{Journal of Computer and System Sciences}, 70(2):128-144,
    2005.



    \bibitem{Huang}
    X. Huang. An Algorithm for Identifying Regions of a DNA
    Sequence that Satisfy a Content Requirement.
    \textit{Computer Applications in the Biosciences}, 10(3):219-225, 1994.


    \bibitem{Kim}
    S. K. Kim. Linear-Time Algorithm for Finding a Maximum-Density Segment of a Sequence.
    \textit{Information Processing Letters}, 86(6):339-342, 2003.


    \bibitem{Lee07}
    D. T. Lee, T.-C. Lin, and H.-I. Lu. Fast Algorithms for the Density finding Problem.
    \textit{Algorithmica}, DOI:10.1007/s00453-007-9023-8, 2007.

    \bibitem{Lin}
    Y.-L. Lin, T. Jiang, and K.-M. Chao. Efficient Algorithms for Locating the Length-Constrained
    Heaviest Segments with Applications to Biomolecular
    Sequence Analysis. \textit{Journal of Computer and System Sciences},
    65(3):570-586, 2002.

    \bibitem{Lipson}
    D. Lipson, Y. Aumann, A. Ben-Dor, N. Linial, and Z. Yakhini.
    Efficient Calculation of Interval Scores for DNA Copy Number Data Analysis.
    \textit{Journal of Computational Biology}, 13(2):215-228, 2006.

    \bibitem{Maragos}
    P. Maragos. Differential Morphology. \textit{Nonlinear
    Image Processing}, 289-329, 2000.

    \bibitem{Moreau}
    J.-J. Moreau. Inf-Convolution, Sous-Additivit$\acute{\mbox{e}}$, Convexit$\acute{\mbox{e}}$ Des Fonctions Num$\acute{\mbox{e}}$riques.
    \textit{Journal de Math$\acute{e}$matiques Pures et Appliqu$\acute{e}$es}, 49:109-154, 1970.




    \bibitem{Rockafellar}
    R. T. Rockafellar. \textit{Convex Analysis}, 1970.


    \bibitem{Stromberg}
    T. Str$\ddot{\mbox{o}}$mberg. The Operation of Infimal Convolution. \textit{Dissertationes Mathematicae}, 352:58, 1996.

    \bibitem{Takaoka}
    T. Takaoka. Efficient Algorithms for the Maximum Subarray Problem by Distance Matrix Multiplication.
    \textit{Electronic Notes in Theoretical Computer Science}, 61:191-200,
    2002.

    \bibitem {Wang03}
    L. Wang and Y. Xu. SEGID: Identifying Interesting Segments in (Multiple)
    Sequence Alignments. \textit{Bioinformatics}, 19(2):297-298,
    2003.
\end{thebibliography}
\end{document}